\documentclass{kms-b}
\newcommand{\publname}{Bull.~Korean Math.~Soc.}
\issueinfo{}
  {}
  {}
  {}
\pagespan{1}{}
\copyrightinfo{}
  {Korean Mathematical Society}

\usepackage{graphicx}
\allowdisplaybreaks

\theoremstyle{plain}
\newtheorem{theorem}{Theorem}[section]

\newtheorem{lemma}[theorem]{Lemma}

\newtheorem{result}{Result}[section]

\theoremstyle{definition}
\newtheorem*{definition}{Definition}
\newtheorem{example}[theorem]{Example}

\theoremstyle{remark}

\def\dil{\Lam}

\def\capset{\Gamma}
\def\dualset{\capset^\ast}

\def\be{\begin{equation}}
\def\ee{\end{equation}}

\def\dm{n}

\def\Rd{\RR^\dm}
\def\Zd{\ZZ^\dm}

\def\Td{\TT^\dm}
\def\CC{\mathbb{C}}

\def\RR{\mathbb{R}}
\def\ZZ{\mathbb{Z}}
\def\TT{\mathbb{T}}

\def\bks{\backslash}

\def\alp{\alpha}                
\def\bet{\beta}
\def\gam{\gamma}                \def\Gam{\Gamma}
\def\del{\delta}

               \def\Lam{\Lambda}
\def\sig{\sigma}                

\def\ome{\omega}                

\def\hat#1{{\widehat {#1}}}

\begin{document}

\title[Wavelet Filter Banks via Extended Laplacian Pyramid Matrices]
{Design of Wavelet Filter Banks for Any Dilation Using Extended Laplacian Pyramid Matrices}

\author[Y. Hur]{Youngmi Hur}
\address{Youngmi Hur \\ Department of Mathematics \\ Yonsei University \\ Seoul 03722,  Korea}
\email{yhur@yonsei.ac.kr}

\author[S. Kim]{Sungjoo Kim}
\address{Sungjoo Kim \\ School of Mathematics and Computing (Mathematics)\\ Yonsei University\\
Seoul 03722,  Korea}
\email{sungjookim@yonsei.ac.kr}
\thanks{The first author was supported in part by the National Research Foundation of Korea (NRF) [Grant Number 2021R1A2C1007598].}



\begin{abstract}
In this paper, we present a new method for designing wavelet filter banks for any dilation matrices and in any dimension. Our approach utilizes extended Laplacian pyramid matrices to achieve this flexibility. By generalizing recent tight wavelet frame construction methods based on the sum of squares representation, we introduce the sum of vanishing products (SVP) condition, which is significantly easier to satisfy. These flexible design methods rely on our main results, which establish the equivalence between the SVP and mixed unitary extension principle conditions. Additionally, we provide illustrative examples to showcase our main findings.
\end{abstract}

\maketitle

\section{Introduction}

Wavelet filter banks have been used in various signal and image processing applications. However, the design of wavelet filter banks, particularly for multidimensional cases and arbitrary dilation matrices, remains challenging. 

The design of wavelet filter banks is closely related to the construction of wavelet frames \cite{RS1,RS2,HanOrtho}, which provide a useful alternative to wavelet bases, as they allow flexibility in construction while preserving many important properties. Recent advances in wavelet frame construction have employed the sum of squares (SOS) representation, originally developed for polynomials, and adapted it for trigonometric polynomials, providing a systematic method to construct tight wavelet frames \cite{LaiStock,AlgPersI}. These methods, while useful, often require solving specific factorization problems involving trigonometric polynomials, which can be quite challenging for general dilation matrices or multidimensions. More precisely, we need to find SOS generators that satisfy the SOS condition as defined in Eq.~(\ref{eq:SOS}) for a given lowpass filter. These construction methods for tight wavelet frames are based on the equivalence between the SOS condition and the unitary extension principle (UEP) condition~\cite{RS1}. 

In this paper, we present an easy method for designing wavelet filter banks for any dilation matrices and in any dimension. The key idea of our design method lies in the use of what we term the sum of vanishing products (SVP) condition (cf. Eq. (\ref{eq:SVP})) for a given pair of lowpass filters, which subsumes the SOS condition as a special case. To this end, we extend and use the Laplacian pyramid matrices studied previously \cite{SLP}. Throughout this paper, and in the process of establishing these extended Laplacian pyramid matrices, we utilize the so-called polyphase representation~\cite{DV2,Hur1} of filters.

The remainder of this paper is organized as follows. Section~\ref{subS:polyphase} reviews some preliminaries, including filters and Laplacian pyramid matrices, particularly in terms of the polyphase representation. Section~\ref{subS:SOSandWF} provides a quick review of wavelet filter bank design and SOS-based tight wavelet frames. Section~\ref{S:main} presents our main results. Section~\ref{subS:SVPandELP} introduces the concept of the sum of vanishing products (SVP) condition and the extended Laplacian pyramid matrices. Section~\ref{subS:SVPeqMUEP} establishes the equivalence between the SVP condition and the mixed unitary extension principle (MUEP) condition~\cite{RS2}; this equivalence is simple yet powerful, allowing the flexible design of wavelet filter banks. Finally, Section~\ref{S:examples} concludes the paper with examples illustrating our main findings.

\section{Preliminaries}
\label{S:prelim}
\subsection{Filters and Laplacian pyramid matrices}
\label{subS:polyphase}
Let \(\dil\) denote the dilation matrix for sampling. That is, \(\dil\) is an
\(\dm\times\dm\) integer matrix whose eigenvalues are strictly larger than \(1\) in magnitude.
Let \(q:=|\det\dil|\). In this paper, we assume that the filter \(f:\Zd\to\RR\) is finitely supported. The filter \(f\) is called {\it highpass} if 
\(
\sum_{m\in\Zd}f(m)=0,
\)
and {\it lowpass} if 
\(
\sum_{m\in\Zd}f(m)=\sqrt{q}=\sqrt{|\det\dil|}.
\)

Let \(\capset\) be the complete set of distinct coset representatives of \(\Zd/\dil\Zd\) containing \(0\), and \(\dualset\) be the complete set of distinct coset representatives of \((2\pi)(((\dil^T)^{-1}\Zd)/\Zd)\) containing \(0\). Here, \(T\) is used for the matrix transpose. We will also use \(\nu_0,\nu_1,\dots,\nu_{q-1}\) and \(\gam_0,\gam_1,\dots,\gam_{q-1}\) to denote the elements of \(\capset\) and \(\dualset\), respectively, so that
\[\capset=\{\nu_0:=0,\nu_1,\dots,\nu_{q-1}\},\quad \dualset=\{\gam_0:=0,\gam_1,\dots,\gam_{q-1}\}.\]
For example, for the 2-D dyadic case, i.e., when \(\dm=2\) and \(\dil=2{\tt I}_2\), the sets \(\capset=\{(0,0),(1,0),(0,1),(1,1)\}\) and \(\dualset=\{(0,0),(\pi,0),(0,\pi),(\pi,\pi)\}\) can be used.
Here, \({\tt I}_m\) is used to denote the \(m\times m\) identity matrix.

For a given dilation matrix \(\dil\), the {\it polyphase decomposition} \cite{V} of a filter \(f\) is a set of \(q\) sub-filters, \(f_\nu\), \(\nu\in\capset\),
that are defined as
\[f_\nu(m):=f(\dil m-\nu), \quad\forall m\in\Zd.\]
Let \(\Td:=\{z=[z_1,\dots,z_\dm]^T\in \CC^\dm: |z_i|=1\}\). We recall that the {\it \(z\)-transform} \(F(z)\) of \(f:\Zd\to\RR\) is defined as
\[F(z):=\sum_{m\in\Zd} f(m) z^{-m},\quad z\in \Td,\]
where, for \(z=[z_1,\dots,z_\dm]^T\) and \(m=[m_1,\dots,m_\dm]^T\in\Zd\), \(z^m:=\prod_{j=1}^\dm z_j^{m_j}\). Note that since \(f\) is assumed to be finitely supported, \(F(z)\) is a Laurent polynomial in \(z\).

It is easy to see that the \(z\)-transform of \(f\) can be written as
\begin{equation}
\label{eq:zpoly}
F(z)=\sum_{\nu\in\capset}z^{\nu}F_\nu(z^\dil)
\end{equation}
where \(F_\nu(z)\) is the \(z\)-transform of \(f_\nu\), and \(z^\dil\) is defined as \([z^{\dil_1},\dots,z^{\dil_\dm}]^T\) with \(\dil_1, \dots, \dil_\dm\) being the column vectors of \(\dil\). The {\it polyphase representation} \cite{DV2,Hur1} of the filter \(f\) is defined as \({\tt F}(z)=[F_{\nu_0}(z),F_{\nu_1}(z),\dots,F_{\nu_{q-1}}(z)]\), which is a row vector of length \(q\).

The Laplacian pyramid algorithm is introduced in \cite{BA}. For a pair of lowpass filters \(h\) and \(g\) with the polyphase representations \({\tt H}(z)\) and \({\tt G}(z)\), respectively, the following \((q+1)\times q\) matrix obtained from the Laplacian pyramid analysis algorithm
\begin{equation*}
\left[\begin{array}{c}
                {\tt H}(z)\\
                {\tt I}_q-{\tt G}(z)^\ast{\tt H}(z)
               \end{array}\right], \quad z\in\Td,
\end{equation*}
is studied in \cite{DV2,Hur1}. Here, \({\tt G}(z)^\ast\) denotes the conjugate transpose of \({\tt G}(z)\), i.e., \({\tt G}(z)^\ast=\overline{{\tt G}(z)}^T={\tt G}(z^{-1})^T\). We refer to the above matrix as the {\it Laplacian pyramid matrix} and denote it by \(\Phi_{h,g}(z)\). 

It is well known that if the two lowpass filters \(h\) and \(g\) are {\it biorthogonal}, i.e., if \({\tt H}(z){\tt G}(z)^\ast=1\), then we have the identity
\begin{equation}
\label{eq:bior}
\Phi_{g,h}(z)^\ast\Phi_{h,g}(z)=\left[\begin{array}{c}
                {\tt G}(z)^\ast, \; {\tt I}_q-{\tt G}(z)^\ast{\tt H}(z)
               \end{array}\right]\Phi_{h,g}(z)={\tt I}_q,\quad \forall z\in\Td,
\end{equation}
and this identity has been used in designing wavelet filter banks~\cite{DV2,Hur1}. 

In the special case when \(g=h\), the Laplacian pyramid matrix becomes
\[\Phi_{h}(z):=\Phi_{h,h}(z)
=\left[\begin{array}{c}
                {\tt H}(z)\\
                {\tt I}_q-{\tt H}(z)^\ast{\tt H}(z)
               \end{array}\right],\]
and if, in addition, \(h\) is orthogonal so that \({\tt H}(z){\tt H}(z)^\ast=1\), then we have
\[\Phi_{h}(z)^\ast\Phi_{h}(z)={\tt I}_q,\quad \forall z\in\Td.\]
In other words, in this case, the Laplacian pyramid matrix \(\Phi_{h}(z)\) satisfies the unitary condition. 
New construction methods for tight wavelet frames have been studied under this setting \cite{DV2}, and in a more general setting without the unitary condition of \(\Phi_{h}(z)\), using the notion of scalable frames \cite{SLP}.

A trigonometric polynomial \(\rho:=\hat{f}\) is called the {\it mask} associated with a filter \(f\) if, for all \(\ome\in\Theta^\dm:=[-\pi,\pi]^\dm\),
\begin{equation}
\label{eq:maskdef}
\rho(\ome):=\hat{f}(\ome):=\frac{1}{\sqrt{|\det \dil|}} F(e^{i\ome})=\frac{1}{\sqrt{|\det \dil|}}\sum_{m\in\Zd} f(m) e^{-i\ome\cdot m},
\end{equation}
where \(F(z)\) is the \(z\)-transform of \(f\).
Note that, given a filter \(f\), the associated mask \(\rho(\ome)=\hat{f}(\ome)\) is the Fourier transform of the filter \(f\). We see that Eq.~(\ref{eq:zpoly}) can also be written as, for all \(\ome \in \Theta^\dm\),
\begin{equation}
\label{eq:omepoly}
F(e^{i\ome})=\sum_{\nu\in\capset}e^{i\ome\cdot\nu}F_\nu(e^{i\dil^T\ome}), \hbox{ or equivalently, } \hat{f}(\ome)=\sum_{\nu\in\capset}e^{i\ome\cdot\nu}\hat{f_\nu}(\dil^T\ome).
\end{equation}

We say the mask \(\rho=\hat{f}\) is highpass (respectively, lowpass) if the associated filter \(f\) is highpass (respectively, lowpass). Thus, the mask \(\rho\) is lowpass if and only if \(\rho(0)=1\), which in turn is equivalent to \(F({\tt 1})=\sqrt{|\det \dil|}\). Here and below, \({\tt 1}:=[1,\cdots,1]^T\in\RR^n\) denotes the column vector of ones.

In this paper, we will assume that a lowpass filter\ \(f\) always satisfies \(\hat{f}(\gam)=0\) for all \(\gam\in\dualset\bks\{0\}\), since almost all lowpass filters used in practice meet this condition.

\subsection{Wavelet filter bank design and SOS-based tight wavelet frames}
\label{subS:SOSandWF}
Let \(h\) be a lowpass filter and \(h_1,\cdots, h_s\)  be highpass filters. We recall that the set of filters \(\{h, h_1,\cdots, h_s\}\) is called a {\it wavelet filter bank}.

For a lowpass filter \(g\) and highpass filters \(h_1^d,\cdots, h_s^d\), we refer to the wavelet filter banks 
\(\{h, h_1,\cdots, h_s\}\) and \(\{g, h_1^d,\cdots, h_s^d\}\) as the {\it primal wavelet filter bank} and the {\it dual wavelet filter bank}, respectively, if 
their Fourier transforms satisfy 
\begin{equation}
\label{eq:origMUEP}
\hat{h}(\ome)\overline{\hat{g}(\ome+\gam)}+\sum_{i=1}^{s}\hat{h}_i(\ome)\overline{\hat{h}_i^d(\ome+\gam)}=\delta_{0\gam},~\forall \gam\in\dualset, \forall\ome \in \Theta^\dm.
\end{equation}
This condition is known as the mixed unitary extension principle (MUEP) condition~\cite{RS2}.
Here and below, \(\del_{ij}=1\) if \(i=j\), and \(\del_{ij}=0\) otherwise. 

It is easy to see that the MUEP condition (\ref{eq:origMUEP}) is equivalent to the matrix identity 
\begin{equation}
\label{eq:matrixMUEP}
\left[\begin{array}{c}
                \overrightarrow{g}(\ome)\\
                \overrightarrow{h_1^d}(\ome)\\
                \cdots\\
                \overrightarrow{h_s^d}(\ome)\\
               \end{array}\right]^\ast
               \left[\begin{array}{c}
                \overrightarrow{h}(\ome)\\
                \overrightarrow{h_1}(\ome)\\
                \cdots\\
                \overrightarrow{h_s}(\ome)\\
               \end{array}\right]={\tt I}_q, 
\end{equation}
where, for a filter \(f\), \(\overrightarrow{f}(\ome)\) is used to denote \([\hat{f}(\ome),\hat{f}(\ome+\gam_1),\dots,\hat{f}(\ome+\gam_{q-1})]\).

In the special case where \(g=h\), and \(h_i^d =h_i\) for all \(i=1,\cdots, s\), which satisfy the unitary extension principle (UEP) condition~\cite{RS1}
\begin{equation}
\label{eq:origUEP}
\hat{h}(\ome)\overline{\hat{h}(\ome+\gam)}+\sum_{i=1}^{s}\hat{h}_i(\ome)\overline{\hat{h}_i(\ome+\gam)}=\delta_{0\gam},~\forall \gam\in\dualset, \forall\ome \in \Theta^\dm,
\end{equation}
we say that \(\{h, h_1,\cdots, h_s\}\) is a {\it tight wavelet filter bank}.

Let \({\tt H}(z)\), \({\tt G}(z)\), \({\tt H}_i(z)\), and \({\tt H}_i^d(z)\) be the polyphase representation of \(h\), \(g\), \(h_i\) and \(h_i^d\), respectively, for \(i=1,\cdots,s\). Then \(\{h, h_1,\cdots, h_s\}\) and \(\{g, h_1^d,\cdots, h_s^d\}\) satisfy the MUEP condition~(\ref{eq:origMUEP}) (or equivalently (\ref{eq:matrixMUEP})) if and only if, for all \(z\in\Td\),
\begin{equation}
\label{eq:polyMUEP}
\left[\begin{array}{c}
                {\tt G}(z)^\ast, \; {\tt H}_1^d(z)^\ast, \; \cdots, \; {\tt H}_s^d(z)^\ast
               \end{array}\right]
               \left[\begin{array}{c}
                {\tt H}(z)\\
                {\tt H}_1(z)\\
                \cdots\\
                {\tt H}_s(z)\\
               \end{array}\right]={\tt I}_q.
\end{equation}         

In fact, to get the matrix version of the MUEP condition in Eq.~(\ref{eq:matrixMUEP}) from the condition (\ref{eq:polyMUEP}), we recall that the Fourier transform matrix defined as
\begin{equation}
\label{eq:FTmatrix}
X(\ome):=\frac{1}{\sqrt{|\\ \det \dil|}}[e^{i(\ome+\gam)\cdot\nu}]_{\nu\in\capset, \gam\in\dualset}
\end{equation}
satisfies \(X(\ome)^\ast X(\ome)={\tt I}_q\), for all \(\ome\in\Theta^\dm\). Using this matrix and Eq. (\ref{eq:omepoly}), we get, for \(k=0,1,\cdots,q-1\),
\[\left({\tt H}(e^{i\Lam^T\ome})X(\ome)\right)_{k}=\frac{1}{\sqrt{\left| \det \dil\right|}}\sum_{l=0}^{q-1}e^{i(\ome+\gam_k)\cdot\nu_l}(H)_{\nu_l}(e^{i\Lam^T(\ome+\gam_k)})=\hat{h}(\ome+\gam_k).\]
As a result, we have
\begin{equation}
\label{eq:HXeqhvec}
{\tt H}(e^{i\Lam^T\ome})X(\ome)=\overrightarrow{h}(\ome),\quad {\tt H}_j(e^{i\Lam^T\ome})X(\ome)=\overrightarrow{h_j}(\ome), \quad  \hbox{ for }j=1,\cdots, s.
\end{equation}
This shows that the condition (\ref{eq:polyMUEP}) implies the condition (\ref{eq:matrixMUEP}), and the converse direction can be obtained similarly.

Let \(\Psi\) be a finite subset of \(L^2(\Rd)\) and let \(X(\Psi)\) be a wavelet system generated by the set \(\Psi\), defined as follows: 
\begin{eqnarray*}
X(\Psi)=\{|\det(\Lam)|^{\frac{j}{2}}\psi(\Lam^j\cdot-m): \psi\in\Psi, j\in\mathbb{Z},\phantom{,}m\in\Zd\}.    
\end{eqnarray*}
A wavelet system \(X(\Psi)\) is called a {\it frame} if there are constants \(C_1,C_2 >0\) such that 
\begin{eqnarray*} C_1\bigr{|}\bigr{|}f\bigr{|}\bigr{|}_{L^2}^2 \le \sum_{x\in X(\Psi)} |\left<f,x\right>|^2 \le C_2\bigr{|}\bigr{|}f\bigr{|}\bigr{|}_{L^2}^2, \quad \forall f \in L^2(\Rd). \end{eqnarray*} 
As a special case, the wavelet system \(X(\Psi)\) is called a {\it tight frame} if, for some constant \(C\),
\begin{eqnarray*} \bigr{|}\bigr{|}f\bigr{|}\bigr{|}_{L^2}^2 = C \sum_{x\in X(\Psi)} |\left<f,x\right>|^2, \quad \forall f \in L^2(\Rd). \end{eqnarray*}
The following well-known result states that the UEP condition (\ref{eq:origUEP}) is sufficient to obtain a tight wavelet frame~\cite{RS1,HanOrtho}, written in a way specifically tailored to our setting.
\begin{result} 
\label{result:UEP}
Suppose \(\{h, h_1,\cdots, h_s\}\) is a tight wavelet filter bank. 
Let \(\phi\in L^2(\Rd)\) be the refinable function associated with the lowpass filter \(h\), in the sense that \(\hat\phi (\dil\ome)=\hat{h}(\ome)\hat\phi(\ome)\), and let \(\Psi:=\{\psi_1,\cdots, \psi_s\}\) be the mother wavelet set, where \(\psi_i\) is defined via \(\hat\psi_i (\dil\ome):=\hat{h_i}(\ome)\hat\phi(\ome)\). 
Then \(X(\Psi)\) is a tight wavelet frame in \(L^2(\Rd)\).
\end{result}

A nonnegative trigonometric polynomial \(\rho\) is said to have the {\it sum of squares (SOS) representation} if there are trigonometric polynomials \(\rho_1,\cdots,\rho_{J}\)  such that
\[\rho(\ome) = \sum_{j=1}^J|\rho_j(\ome)|^2, \quad \forall \ome \in \Theta^\dm.\]
For the wavelet construction, the relevant SOS problem \cite{LaiStock,AlgPersI} is for the case when \(\rho(\ome)\) is specifically chosen as
\[\rho(\ome):=1-\sum_{\gam\in{\dualset}}|\hat{h}(\ome+\gam)|^2,\]
for a lowpass filter \(h\). 
Equivalently, this SOS problem can be written via the \(z\)-transform \(H_\nu(z)\) of the polyphase decomposition \(h_\nu\), \(\nu\in\capset\), of the lowpass filter \(h\) as follows: there exist Laurent polynomials \(p_j\), \(j=1,\cdots,J'\) such that
\[1-\sum_{\nu\in{\capset}}|H_\nu(z)|^2=\sum_{j=1}^{J'}|p_j(z)|^2, \quad \forall z\in\Td.\]

The exact process of obtaining this condition for \(z\in\Td\) from the relevant SOS problem for \(\ome\in\Theta^\dm\) mentioned above is well known (see, for example,~\cite{LaiStock,AlgPersI}). We generalize this process in the next section (cf. Lemma~\ref{lem:svpinome}) and use it to prove one of our main results, Theorem~\ref{thm:muepthensvp}. 

The following result from ~\cite{LaiStock,AlgPersI} is about a tight wavelet filter bank design using the SOS representation, stated in a way adapted to our setting. 
\begin{result} 
\label{result:SOStoUEP} 
Let \(h\) be a lowpass filter, and let \({\tt H}(z)=[H_{\nu_0}(z),\dots,H_{\nu_{q-1}}(z)]\) be its polyphase representation.
Suppose that \(h\) is not orthogonal, i.e., \({\tt H}(z){\tt H}(z)^\ast\ne 1\), and that  there exist nonzero Laurent polynomials \(p_1,\cdots,p_{J}\) such that
\begin{equation}
\label{eq:SOS}
1-{\tt H}(z){\tt H}(z)^\ast=\sum_{j=1}^{J}|p_j(z)|^2,\quad z\in\Td.
\end{equation} 
Define highpass filters \(h_1,\cdots, h_{J+q}\) as follows:
\begin{eqnarray*}
\hat{h_j}(\ome)&:=&\hat{h} (\ome)\overline{{p_j}(e^{i\dil^T\ome})},\quad j=1,\cdots, J,\\
\hat{h_{J+1+m}}(\ome)&:=&q^{-1/2}e^{i\nu_m\cdot\ome}- \hat{h} (\ome)\overline{H_{\nu_m}(e^{i\dil^T\ome})},\quad m=0,\cdots, q-1.
\end{eqnarray*}
Then \(\{h, h_1,\cdots, h_{J+q}\}\) forms a tight wavelet filter bank.
\end{result}

In this paper, we refer to the above condition in Eq.~(\ref{eq:SOS}) as the {\it SOS condition} for the lowpass filter \(h\), and to the functions \(p_1,\cdots,p_{J}\) as the {\it SOS generators}.
By combining this result with Result~\ref{result:UEP}, we see that if a lowpass filter \(h\) satisfies the SOS condition (\ref{eq:SOS}), then it gives rise to a tight wavelet frame in \(L^2(\Rd)\).

Note that for the filter \(h\) to satisfy the SOS condition, it is necessary that
\begin{equation*}
1-{\tt H}(z){\tt H}(z)^\ast\ge 0,\quad z\in\Td,
\end{equation*}
which is known as the sub-QMF condition for the filter \(h\) in the literature.

\section{Wavelet filter banks using extended Laplacian pyramid matrices}
\label{S:main}
\subsection{The SVP condition and extended Laplacian pyramid matrices}
\label{subS:SVPandELP}
\begin{definition}
Let \(h\) and \(g\) be lowpass filters with the polyphase representations \({\tt H}(z)\) and \({\tt G}(z)\), respectively.
We say that the lowpass filters \(h\) and \(g\) satisfy the sum of vanishing products (SVP) condition if there exist Laurent polynomials \(k_j\) and \(l_j\) satisfying \(k_j({\tt 1})=l_j({\tt 1})=0\), for all \(1\le j\le J\), such that
\begin{equation}
\label{eq:SVP}
1-{\tt H}(z){\tt G}(z)^\ast=\sum_{j=1}^{J} k_j(z)\overline{l_j(z)},\quad \forall z\in\Td.
\end{equation}
\end{definition}

We refer to the above functions \(k_1,\cdots,k_{J}\) and \(l_1,\cdots,l_{J}\) as the {\it SVP generators}.

Clearly, if a lowpass filter \(h\) satisfies the SOS condition (\ref{eq:SOS}) with \(p_j\) as SOS generators, then we get the SVP condition (\ref{eq:SVP}), with \(g\) chosen the same as \(h\), and \(p_j\) as the SVP generators.

\begin{definition}
Let \(h\) and \(g\) be lowpass filters with the polyphase representations \({\tt H}(z)\) and \({\tt G}(z)\), respectively.
Let \(l_j\), \(j=1,\cdots,J\), be the Laurent polynomials with \(l_j({\tt 1})=0\). 
Then we define the {\it extended Laplacian pyramid matrix} \(\Phi_{h,g, [l_1,\cdots,l_{J}]}(z)\) as the following \((q+J+1)\times q\) matrix:
\begin{equation*}
\Phi_{h,g, [l_1,\cdots,l_{J}]}(z)
:=\left[\begin{array}{c}
                {\tt H}(z)\\
               \overline{l_1(z)}{\tt H}(z)\\
               \vdots\\
               \overline{l_J(z)}{\tt H}(z)\\
                {\tt I}_q-{\tt G}(z)^\ast{\tt H}(z)
               \end{array}\right].
\end{equation*} 

\end{definition}
The following simple lemma is a key to our design method in this paper.

\begin{lemma}
\label{lem:core}
Let \(h\) and \(g\) be lowpass filters that satisfy the SVP condition (\ref{eq:SVP}) with generators \(k_1,\cdots,k_{J}\) and \(l_1,\cdots,l_{J}\).
Then the extended Laplacian pyramid matrices \(\Phi_{h,g, [l_1,\cdots,l_{J}]}(z)\) and \(\Phi_{g,h, [k_1,\cdots,k_{J}]}(z)\) satisfy the identity
\begin{equation}
\label{eq:core}
\Phi_{g,h, [k_1,\cdots,k_{J}]}(z)^\ast\Phi_{h,g, [l_1,\cdots,l_{J}]}(z)={\tt I}_q, \quad\forall z\in\Td.
\end{equation}
\end{lemma}

\begin{proof}
From (\ref{eq:SVP}), we see that 
\begin{eqnarray*}
[k_1(z),\cdots,k_{J}(z)][l_1(z),\cdots,l_{J}(z)]^\ast=1-{\tt H}(z){\tt G}(z)^\ast .\end{eqnarray*} Using this identity, we have
\begin{eqnarray*}
&{\;}&\Phi_{g,h, [k_1,\cdots,k_{J}]}(z)^\ast\Phi_{h,g, [l_1,\cdots,l_{J}]}(z)\\
&=&{\tt G}(z)^\ast{\tt H}(z)+{\tt G}(z)^\ast([k_1(z),\cdots,k_{J}(z)][l_1(z),\cdots,l_{J}(z)]^\ast){\tt H}(z)\\
&+&\left({\tt I}_q-{\tt G}(z)^\ast{\tt H}(z)\right)\left({\tt I}_q-{\tt G}(z)^\ast{\tt H}(z)\right)\\
&=&{\tt G}(z)^\ast{\tt H}(z)+{\tt G}(z)^\ast(1-{\tt H}(z){\tt G}(z)^\ast){\tt H}(z)\\
&+&\left({\tt I}_q-{\tt G}(z)^\ast{\tt H}(z)\right)\left({\tt I}_q-{\tt G}(z)^\ast{\tt H}(z)\right)\\
&=&{\tt G}(z)^\ast\left(2-{\tt H}(z){\tt G}(z)^\ast-2+{\tt H}(z){\tt G}(z)^\ast\right){\tt H}(z)+{\tt I}_q={\tt I}_q,
\end{eqnarray*}
which completes the proof.
\end{proof}

Note that if \(h\) and \(g\) are biorthogonal, i.e., \(1-{\tt H}(z){\tt G}(z)^\ast=0\), then the SVP condition (\ref{eq:SVP}) holds true trivially, with all the generators chosen as the zero Laurent polynomial. In this case, the identity in Eq.~(\ref{eq:core}) reduces to the previous identity in Eq.~(\ref{eq:bior}). Therefore, in this paper, we are primarily interested in the wavelet filter bank design when the two lowpass filters that we start with are not biorthogonal.

\subsection{Equivalence between the SVP and MUEP conditions}
\label{subS:SVPeqMUEP}
In this subsection, we show the equivalence between the SVP condition (\ref{eq:SVP}) and the MUEP condition (\ref{eq:origMUEP}). 
We start with the direction of deriving the MUEP condition from the SVP condition.
\begin{theorem}
\label{thm:main}
Let \(h\), \(g\) be lowpass filters, and let \({\tt H}(z)=[H_{\nu_0}(z),\dots,H_{\nu_{q-1}}(z)]\) and \({\tt G}(z)=[G_{\nu_0}(z),\dots,G_{\nu_{q-1}}(z)]\) be their polyphase representations, respectively. Suppose that \(h\) and \(g\) are not biorthogonal, i.e., \({\tt H}(z){\tt G}(z)^\ast\ne 1\), and that  there exist nonzero generators \(k_1,\cdots,k_{J}\) and \(l_1,\cdots,l_{J}\) for the SVP condition (\ref{eq:SVP}).
Define highpass filters \(h_1,\cdots, h_{J+q}\) as follows:
\begin{eqnarray}
\label{eq:hj}
\hat{h_j}(\ome)&:=&\hat{h} (\ome)\overline{{l_j}(e^{i\dil^T\ome})},\quad j=1,\cdots, J,\\
\hat{h_{J+1+m}}(\ome)&:=&q^{-1/2}e^{i\nu_m\cdot\ome}- \hat{h} (\ome)\overline{G_{\nu_m}(e^{i\dil^T\ome})},\quad m=0,\cdots, q-1,
\label{eq:hm}
\end{eqnarray}
and highpass filters \(h_1^d,\cdots, h_{J+q}^d\) as follows:
\begin{eqnarray*}
\hat{h_j^d}(\ome)&:=&\hat{g} (\ome)\overline{{k_j}(e^{i\dil^T\ome})},\quad j=1,\cdots, J,\\
\hat{h_{J+1+m}^d}(\ome)&:=&q^{-1/2}e^{i\nu_m\cdot\ome}- \hat{g} (\ome)\overline{H_{\nu_m}(e^{i\dil^T\ome})},\quad m=0,\cdots, q-1.
\end{eqnarray*}
Then \(\{h, h_1,\cdots, h_{J+q}\}\) and \(\{g, h_1^d,\cdots, h_{J+q}^d\}\) are the primal and dual wavelet filter banks, respectively. 
\end{theorem}

\begin{proof}
We first note that by Lemma~\ref{lem:core}, Eq.~(\ref{eq:core}) holds, i.e., we have 
\begin{eqnarray*}\Phi_{g,h, [k_1,\cdots,k_{J}]}(z)^\ast\Phi_{h,g, [l_1,\cdots,l_{J}]}(z)={\tt I}_q, \quad\forall z\in\Td.\end{eqnarray*}

Let \(\mathcal{K}:=[k_1,\cdots,k_{J}]\) and \(\mathcal{L}:=[l_1,\cdots,l_{J}]\). 
Then, by using the Fourier transform matrix \(X(\ome)\) in Eq.~(\ref{eq:FTmatrix}), we get 
\begin{equation*}
\left(\Phi_{g,h, \mathcal{K}}(e^{i\dil^T\ome})X(\ome)\right)^\ast\left(\Phi_{h,g, \mathcal{L}}(e^{i\dil^T\ome})X(\ome)\right)={\tt I}_q,\quad \forall \ome\in\Theta^\dm.
\end{equation*}
By recalling that \( {\tt H}(e^{i\dil^T\ome})X(\ome)=\overrightarrow{h}(\ome)=[\hat{h}(\ome),\hat{h}(\ome+\gam_1),\dots,\hat{h}(\ome+\gam_{q-1})]\) from Eq. (\ref{eq:HXeqhvec}), we see that \(\Phi_{h,g, \mathcal{L}}(e^{i\dil^T\ome})X(\ome)\) reduces to 
\[		\left[\begin{array}{c}
                {\tt H}(e^{i\dil^T\ome})\\
                \mathcal{L}(e^{i\dil^T\ome})^\ast{\tt H}(e^{i\dil^T\ome})\\
                {\tt I}_q-{\tt G}(e^{i\dil^T\ome})^\ast{\tt H}(e^{i\dil^T\ome})
               \end{array}\right]X(\ome)
               =
               \left[\begin{array}{c}
               \overrightarrow{h}(\ome)\\
                \mathcal{L}(e^{i\dil^T\ome})^\ast\overrightarrow{h}(\ome)\\
                X(\ome)-{\tt G}(e^{i\dil^T\ome})^\ast\overrightarrow{h}(\ome)
               \end{array}\right].\]
Computing \(\mathcal{L}(e^{i\dil^T\ome})^\ast\overrightarrow{h}(\ome)\) and \(X(\ome)-{\tt G}(e^{i\dil^T\ome})^\ast\overrightarrow{h}(\ome)\) give
\[            \mathcal{L}(e^{i\dil^T\ome})^\ast\overrightarrow{h}(\ome)
               =
               \left[\begin{array}{c}
               \overrightarrow{h}(\ome) \overline{l_1(e^{i\dil^T\ome})}\\
                \vdots\\
                \overrightarrow{h}(\ome)\overline{l_J(e^{i\dil^T\ome})}
               \end{array}\right]               
		=
               \left[\begin{array}{c}
                \overrightarrow{h_1}(\ome)\\
                \vdots\\
                \overrightarrow{h_J}(\ome)
               \end{array}\right]               
\]
and
\[
		X(\ome)-{\tt G}(e^{i\dil^T\ome})^\ast\overrightarrow{h}(\ome)
		=
               X(\ome)-
               \left[\begin{array}{c}
               \overrightarrow{h}(\ome) \overline{G_{\nu_0}(e^{i\dil^T\ome})}\\
                \vdots\\
                \overrightarrow{h}(\ome)\overline{G_{\nu_{q-1}}(e^{i\dil^T\ome})}\\
               \end{array}\right]               		
		=
               \left[\begin{array}{c}
                \overrightarrow{h_{J+1}}(\ome)\\
                \vdots\\
                \overrightarrow{h_{J+q}}(\ome)
               \end{array}\right],               
\]
where \(h_1,\cdots, h_{J+q}\) are the highpass filters defined in (\ref{eq:hj}) and (\ref{eq:hm}). Thus, \(\Phi_{h,g, \mathcal{L}}(e^{i\dil^T\ome})X(\ome)\) corresponds to the right matrix on the left-hand side of the identity in Eq.~(\ref{eq:matrixMUEP}) with \(s:=J+q\).

Similar observation can be made for the dual filters \(g, h_1^d,\cdots, h_{J+q}^d\) to show that \((\Phi_{g,h, \mathcal{K}}(e^{i\dil^T\ome})X(\ome))^\ast\) corresponds to the left matrix on the left-hand side of the identity in Eq.~(\ref{eq:matrixMUEP}) with \(s:=J+q\).
Therefore, \(\{h, h_1,\cdots, h_{J+q}\}\) and \(\{g, h_1^d,\cdots, h_{J+q}^d\}\) satisfy the matrix version of MUEP condition in Eq.~(\ref{eq:matrixMUEP}) and this completes the proof.
\end{proof}

Our theorem here is a generalization of Result \ref{result:SOStoUEP} discussed in Section \ref{subS:SOSandWF}, as we obtain Result \ref{result:SOStoUEP} as a special case by considering when \(h\) satisfies the SOS condition (\ref{eq:SOS}).

In our next theorem, we show that the MUEP condition implies the SVP condition. We use the following lemma to prove our theorem. The proof of this lemma can be obtained by slightly generalizing the arguments used in proving Lemma 1 in \cite{HLO}.

\begin{lemma}
\label{lem:svpinome}
Let \(h\) and \(g\) be lowpass filters. Suppose that there exist trigonometric polynomials \(\rho_1, \cdots, \rho_J\) and \(\mu_1,\cdots, \mu_J\) such that
\begin{equation}
\label{eq:svpinome}
1-\sum_{\gam\in{\dualset}}\hat{h}(\ome+\gam)\overline{\hat{g}(\ome+\gam)}=\sum_{j=1}^J \rho_j(\ome)\overline{\mu_j(\ome)},\quad \ome \in \Theta^\dm,
\end{equation}
with \(\sum_{\gam\in\dualset}e^{-i\gam\cdot\nu}\rho_j(\gam)=\sum_{\gam\in\dualset}e^{-i\gam\cdot\nu}\mu_j(\gam)=0\), for all  \(j=1,\cdots,J\) and \(\nu\in\capset\).
Then \(h\) and \(g\) satisfy the SVP condition (\ref{eq:SVP}) with \(J\times |\det\dil|\) pairs of generators.
\end{lemma}

\begin{proof} 
Let \(\rho(\ome)=1-\sum_{\gam\in{\dualset}}\hat{h}(\ome+\gam)\overline{\hat{g}(\ome+\gam)}\). Note that since the lowpass filters \(h\) and \(g\) are assumed to satisfy \(\hat{h}(0)=\hat{g}(0)=1\), and \(\hat{h}(\gam)=\hat{g}(\gam)=0\) for all \(\gam\in\dualset\bks\{0\}\), we have \(\rho(0)=0\).

Since \(\overrightarrow{h} (\ome)=[\hat{h}(\ome+\gam_0),\dots,\hat{h}(\ome+\gam_{q-1})]\) and \(\overrightarrow{g} (\ome)=[\hat{g}(\ome+\gam_0),\dots,\hat{g}(\ome+\gam_{q-1})]\) by definition, we get \(\rho(\ome)=1-\overrightarrow{h} (\ome)  \overrightarrow{g} (\ome)^\ast\).
Using Eq.~(\ref{eq:HXeqhvec}) for each of the filters \(h\) and \(g\), we see that
\[\rho(\ome)=1-{\tt H}(e^{i\dil^T\ome})X(\ome)\left({\tt G}(e^{i\dil^T\ome})X(\ome)\right)^\ast=1-{\tt H}(e^{i\dil^T\ome}){\tt G}(e^{i\dil^T\ome})^\ast,\]
where \({\tt H}(z)\) and \({\tt G}(z)\) are the polyphase representations of \(h\) and \(g\), respectively, and \(X(\ome)\) is the Fourier transform matrix in Eq.~(\ref{eq:FTmatrix}).

Furthermore, since \(\rho(\ome+\gam)=\rho(\ome)\) for all \(\gam\in\dualset\), the given identity (\ref{eq:svpinome}) imples 
\[\rho(\ome)=\frac{1}{|\det\dil|}\sum_{\gam\in{\dualset}}\rho(\ome+\gam)=\frac{1}{|\det\dil|}\sum_{j=1}^J \sum_{\gam\in{\dualset}}\rho_j(\ome+\gam)\overline{\mu_j(\ome+\gam)}.\]
Let \(r_j\) and \(u_j\) be the filters associated with masks \(\rho_j\) and \(\mu_j\), respectively, so that \(\rho_j(\ome)=\hat{r_j}(\ome)\) and \(\mu_j(\ome)=\hat{u_j}(\ome)\) (cf. Eq.~(\ref{eq:maskdef})).
Using Eq.~(\ref{eq:omepoly}) for each of the filters \(r_j\) and \(u_j\), we see that
\begin{eqnarray*}
&&\sum_{\gam\in{\dualset}}\hat{r_j}(\ome+\gam)\overline{\hat{u_j}(\ome+\gam)}\\&=&\sum_{\gam\in{\dualset}}\sum_{\nu\in\capset}e^{i(\ome+\gam)\cdot\nu}\hat{(r_j)_\nu}(\dil^T\ome)\overline{\sum_{\nu'\in\capset}e^{i(\ome+\gam)\cdot\nu'}\hat{(u_j)_{\nu'}}(\dil^T\ome)}\\
&=&\sum_{\nu,\nu'\in\capset} \left(\sum_{\gamma\in\dualset}e^{i\gam\cdot(\nu-\nu')}\right)e^{i\ome\cdot(\nu-\nu')}\hat{(r_j)_\nu}(\dil^T\ome)\overline{\hat{(u_j)_{\nu'}}(\dil^T\ome)}\\
&=&|\det\dil|\sum_{\nu\in\capset}\hat{(r_j)_\nu}(\dil^T\ome)\overline{\hat{(u_j)_{\nu}}(\dil^T\ome)},
\end{eqnarray*}
where \((r_j)_\nu\) and \((u_j)_\nu\) are the polyphase decomposition of \(r_j\) and \(u_j\), respectively, and for the last equality, the following simple identity~\cite{Hur1}
\begin{equation}
\label{eq:simple}
\sum_{\gamma\in\dualset}e^{i\gam\cdot(\nu-\nu')}=|\det\dil|\delta_{\nu\nu'}
\end{equation}
is used.

Hence, the given identity (\ref{eq:svpinome}) becomes
\[1-{\tt H}(e^{i\dil^T\ome}){\tt G}(e^{i\dil^T\ome})^\ast=\sum_{j=1}^J\sum_{\nu\in\capset}\hat{(r_j)_\nu}(\dil^T\ome)\overline{\hat{(u_j)_{\nu}}(\dil^T\ome)}.\]
By setting \(z:=e^{i\dil^T\ome}\), we see that 
\[1-{\tt H}(z){\tt G}(z)^\ast=\sum_{j=1}^J\sum_{\nu\in\capset}k_{j,\nu}(z)\overline{l_{j,\nu}(z)},\]
where \(k_{j,\nu}(e^{i\dil^T\ome}):=\hat{(r_j)_\nu}(\dil^T\ome)\) and \(l_{j,\nu}(e^{i\dil^T\ome}):=\hat{(u_j)_\nu}(\dil^T\ome)\).

We now check the vanishing condition of the generators \(k_{j,\nu}\) and \(l_{j,\nu}\), for all \(j=1,\cdots,J\) and \(\nu\in\capset\).

Let us fix \(\nu\in\capset\). Note that for a filter \(f\), by invoking Eq.~(\ref{eq:omepoly}) again, we have  
\begin{eqnarray}
\sum_{\gam\in\dualset}(e^{i(\ome+\gam)})^{-\nu}\hat{f}(\ome+\gam)&=&\sum_{\gam\in\dualset}(e^{i(\ome+\gam)})^{-\nu}\sum_{\nu'\in\capset}e^{i(\ome+\gam)\cdot{\nu'}}\hat{f_{\nu'}}(\dil^T\ome)\nonumber\\
&=&\sum_{\nu'\in\capset}\left(\sum_{\gam\in\dualset}e^{i\gam\cdot(\nu'-\nu)}\right)e^{i\ome\cdot(\nu'-\nu)}\hat{f_{\nu'}}(\dil^T\ome)\nonumber\\
&=&|\det\dil|\hat{f_{\nu}}(\dil^T\ome),
\label{eq:fnu}
\end{eqnarray}
where, for the last equality, the identity in Eq.~(\ref{eq:simple}) is used. 

Let \(1\le j\le J\) be fixed. Using the identity in Eq.~(\ref{eq:fnu}) for each of the filters \(r_j\) and \(u_j\), since \(\hat{r_j}=\rho_j\) and \(\hat{u_j}=\mu_j\), we see that the given assumptions on \(\rho_j\) and \(\mu_j\) imply

\begin{eqnarray*}\hat{(r_j)_\nu}(0)=\frac{1}{|\det\dil|}\sum_{\gam\in\dualset}e^{-i\gam\cdot\nu}\rho_j(\gam)=0,\end{eqnarray*}
\begin{eqnarray*}\quad \hat{(u_j)_\nu}(0)=\frac{1}{|\det\dil|}\sum_{\gam\in\dualset}e^{-i\gam\cdot\nu}\mu_j(\gam)=0.\end{eqnarray*}
Thus, we have \(k_{j,\nu}({\tt 1})=\hat{(r_j)_\nu}(0)=0\) and \(l_{j,\nu}({\tt 1})=\hat{(u_j)_\nu}(0)=0\).

Therefore, we conclude that \(k_{j,\nu}({\tt 1})=l_{j,\nu}({\tt 1})=0\) for all  \(j=1,\cdots,J\) and \(\nu\in\capset\), and this finishes the proof.
\end{proof}

We are ready to present our result for getting the SVP condition from the MUEP condition. Our proof of this theorem mostly follows the approach presented in ~\cite{AlgPersI}, which is used to show that the UEP condition implies the SOS condition.

\begin{theorem}
\label{thm:muepthensvp}
Let \(h\) and \(g\) be lowpass filters, and let \(\{h, h_1,\cdots, h_{s}\}\) and \(\{g, h_1^d,\cdots, h_{s}^d\}\) be the primal and dual wavelet filter banks, respectively.
Then, \(h\) and \(g\) satisfy the SVP condition.
\end{theorem}
\begin{proof} 
Note that since \(\{h, h_1,\cdots, h_{s}\}\) and \(\{g, h_1^d,\cdots, h_{s}^d\}\) are the primal and dual wavelet filter banks, they satisfy the MUEP condition (\ref{eq:origMUEP}), or equivalently (\ref{eq:matrixMUEP}).

Then, the identity in Eq.~(\ref{eq:matrixMUEP}) can be written as 
\[{\tt I}_q -\overrightarrow{g}(\ome)^\ast \overrightarrow{h}(\ome)=\left[\begin{array}{c}
                \overrightarrow{h_1^d}(\ome)\\
                \cdots\\
                \overrightarrow{h_s^d}(\ome)\\
               \end{array}\right]^\ast
               \left[\begin{array}{c}
                \overrightarrow{h_1}(\ome)\\
                \cdots\\
                \overrightarrow{h_s}(\ome)\\
               \end{array}\right]=:V (\ome)^\ast U (\ome).\]
Hence, by taking the determinant of matrices on both sides, we see that 
\[1-\overrightarrow{h} (\ome)  \overrightarrow{g} (\ome)^\ast=1-\sum_{\gam\in{\dualset}}\hat{h}(\ome+\gam)\overline{\hat{g}(\ome+\gam)}=\det(V(\ome)^\ast U (\ome)).\] 

By Lemma~\ref{lem:svpinome}, it suffices to show there exist trigonometric polynomials \(\rho_j\) and \(\mu_j\), \(j=1,\cdots,J\), such that 
\begin{equation}
\label{eq:suffices}
\det(V(\ome)^\ast U (\ome))=\sum_{j=1}^J \rho_j(\ome)\overline{\mu_j(\ome)},
\end{equation}
where
\begin{equation}
\label{eq:extra}
\sum_{\gam\in\dualset}e^{-i\gam\cdot\nu}\rho_j(\gam)=\sum_{\gam\in\dualset}e^{-i\gam\cdot\nu}\mu_j(\gam)=0,\quad \forall \nu\in\capset.
\end{equation}
Note that both \(U(\ome)\) and \(V(\ome)\) are of the size \(s\times q\). We consider three different cases depending on the relative size of \(s\) and \(q\).

{\noindent Case I  (When \( s<q\)): } In this case, \(\det(V(\ome)^\ast U (\ome))=0\), hence the conditions (\ref{eq:suffices}) and (\ref{eq:extra}) are trivially satisfied by taking the zero trigonometric polynomial as generators.

{\noindent Case II  (When \( s=q\)): } Since \(\det(V(\ome)^\ast U (\ome))=\det(U (\ome))\overline{\det(V(\ome))}\), the condition (\ref{eq:suffices}) is satisfied with \(J=1\) in this case.

Let \(\gam\in\dualset\) be arbitrary but fixed. Then, there exists \(\gam'\in\dualset\) such that \(\gam+\gam'\equiv 0\) \(({\rm mod}\, 2\pi\Zd)\).
Since \(h_1,\cdots,h_s\) are highpass filters, the \(\gam'\)-column of \(U(\ome)\) at \(\ome=\gam\), consistsing of the entries \(h_1(\gam+\gam'),\cdots,h_s(\gam+\gam')\), is euqal to the zero vector.
Therefore, we see that \(\det(U(\gam))=0\) for all \(\gam\in\dualset\).

By using a similar argument for the matrix \(V(\ome)\) considering the highpass filters \(h_1^d,\cdots,h_s^d\) this time, we have \(\det(V(\gam))=0\) for all \(\gam\in\dualset\). This completes the proof since the condition (\ref{eq:extra}) is satisfied for the generating trigonometric polynomials \(\det(U(\ome))\) and  \(\det(V(\ome))\).

{\noindent Case III  (When \(s>q\)): } Note that in this case, by the Cauchy-Binet formula~\cite{MA}, 
\begin{equation*} 
\det(V(\ome)^\ast U(\ome))=\sum_{\sigma\subset \{1,\cdots,s\},|\sigma|=q} \det(U(\ome)_{[\sigma]})\overline{\det(V(\ome)_{[\sigma]})}, 
\end{equation*} 
where \(\sigma\) denotes the permutation of a subset of  \(\{1,\cdots,s\}\) with size \(q\), and the \(q\times q\) square matrices \(U(\ome)_{[\sigma]}\) and \(V(\ome)_{[\sigma]}\) are defined as
\begin{equation*}
U(\ome)_{[\sigma]}:=\left[\begin{array}{c}
                \overrightarrow{h_{\sig(1)}}(\ome)\\
                \vdots\\
                \overrightarrow{h_{\sig(q)}}(\ome)
               \end{array}\right],\quad 
V(\ome)_{[\sigma]}:=\left[\begin{array}{c}
                \overrightarrow{h_{\sig(1)}^d}(\ome)\\
                \vdots\\
                \overrightarrow{h_{\sig(q)}^d}(\ome)
               \end{array}\right].            
\end{equation*}
Thus, the condition (\ref{eq:suffices}) is satisfied with \(J=s!/((s-q)!q!)\).

For the condition (\ref{eq:extra}), we follow a similar process as in Case II above with \(s=q\), but now applied to the matrices \(U(\ome)_{[\sigma]}\) and \(V(\ome)_{[\sigma]}\).
In particular, from the fact that \(h_{\sig(1)},\cdots,h_{\sig(q)}\) and \(h_{\sig(1)}^d,\cdots,h_{\sig(1)}^d\) are highpass filters, we have \(\det(U(\gam)_{[\sigma]})=\det(V(\gam)_{[\sigma]})=0\), for all \(\gam\in\dualset\). This implies the condition (\ref{eq:extra}) for the trigonometric polynomials \(\det(U(\ome)_{[\sigma]})\) and \(\det(V(\ome)_{[\sigma]})\), for all the permutataion \(\sigma\) of  \(\{1,\cdots,s\}\) with size \(q\), which finishes the proof.
\end{proof}

\section{Examples}
\label{S:examples}
In this section, we present some examples to illustrate the application of our proposed wavelet filter bank design method. These examples demonstrate the versatility and effectiveness of our approach in different contexts and dimensions. For simplicity, in all of our examples below, we consider the case  \(g=h\), i.e., the two lowpass filters are the same.

\begin{example} 
\label{ex:first}
Consider the two-dimensional dyadic case, i.e., \(\dil=2{\tt I}_2\).
We choose \(\capset=\{(0,0),(1,0), (0,1),(1,1)\}\) and \(\dualset=\{(0, 0),(\pi,0), (0,\pi),(\pi,\pi)\}\).
Let \(R\) be a one-dimensional trigonometric polynomial introduced in \cite{BA} of the form
\[R(\ome)=a+\cos(\ome)+(1-a)\cos(2\ome),\quad \ome \in \Theta,\]
for \(a\in \RR\). Let \(h\) be the two-dimensional lowpass filter associated with the mask
\[\hat{h}(\ome)=\frac{1}{4}(-2+R(\ome_1)+R(\ome_2)+R(\ome_1+\ome_2)), \quad\ome=(\ome_1,\ome_2)\in\Theta^2.\]
The filter \(h\), obtained by applying the coset sum method \cite{YFC} to the 1-D mask 
\(R\), is studied in Example 4 of \cite{HurLubb1}.
Letting \(g=h\), we show below that the SVP condition for the filter \(h\) is satisfied in this case.

Since the \(z\)-transform of the filter \(h\) is given as 
\[H(z)=\frac{3a}{2}-1+\frac{(1-a)}{4}\sum_{m\in \Gam'} \left({z^{2m} + z^{-2m}}\right)+\sum_{m\in \Gam'} z^m\left({\frac{1 + z^{-2m}}{4}}\right),\]
where \(\Gam':=\capset\bks\{0\}\), we can read off its polyphase representation as follows (cf.~(\ref{eq:zpoly})):
\begin{equation*}
{\tt{H}}(z)=\left[\frac{3a}{2}-1+\displaystyle{\frac{(1-a)}{4}\sum_{m\in \Gam'} \left({z^{m} + z^{-m}}\right) },{\,}\frac{1+z_1^{-1}}{4},{\,}\frac{1+z_2^{-1}}{4},{\,}\frac{1+z_1^{-1}z_2^{-1}}{4} \right]. 
\end{equation*} 
Then, one can write \(1-{\tt{H}}(z){\tt{H}}(z)^{\ast}\) as
\begin{eqnarray*}
&{\,}&\frac{1}{4}-\left(\frac{3a}{2}-1+\frac{(1-a)}{4} \sum_{m\in \Gam'}\left({z^{m} + z^{-m}}\right)\right)^2 \\&+&\sum_{m\in \Gam'}\left(\frac{1}{4}-\frac{1}{16}(z^m+z^{-m}+2)\right)\quad\\
&=&\left(\frac{3}{2}-\frac{3a}{2}-\frac{(1-a)}{4} \sum_{m\in \Gam'}\left({z^{m} + z^{-m}}\right)\right)\\&\cdot&\left(\frac{3a}{2}-\frac{1}{2}+\frac{(1-a)}{4} \sum_{m\in \Gam'}\left({z^{m} + z^{-m}}\right)\right)\\&+&\frac{1}{16}\left(\sum_{m\in \Gam'}(1-z^m)(1-z^{-m})\right)\\
&=&\frac{1-a}{4}\left(\sum_{m\in \Gam'}(1-z^m)(1-z^{-m})\right)E(z)\\&+&\frac{1}{16}\left(\sum_{m\in \Gam'}(1-z^m)(1-z^{-m})\right),
\end{eqnarray*}
where 
\begin{equation*}
E(z)=\frac{3a}{2}-\frac{1}{2}+\frac{(1-a)}{4} \sum_{m\in \Gam'}\left({z^{m} + z^{-m}}\right).
\end{equation*}

The SVP condition (\ref{eq:SVP}) is satisfied with \(g:=h\) in this case, since we have the following:
\begin{equation*}
\displaystyle 1-{\tt H}(z){\tt H}(z)^\ast=\sum_{j=1}^{6} k_j(z)\overline{l_j(z)}=:\mathcal{K}(z)\mathcal{L}(z)^\ast,    
\end{equation*}
where \(\mathcal{K}(z)\) and \(\mathcal{L}(z)\) are given as 
\begin{eqnarray*}
\mathcal{K}(z)&=&\left[(1-a)(1-z_1)E(z),(1-a)(1-z_2)E(z)),(1-a)(1-z_1z_2)E(z),\right.\\&{\,}&\left.\frac{(1-z_1)}{4},\frac{(1-z_2)}{4},\frac{(1-z_1z_2)}{4}\right],
\\\mathcal{L}(z)&=&\frac{1}{4}\biggr[(1-z_1), (1-z_2),(1-z_1 z_2),(1-z_1),(1-z_2),(1-z_1 z_2)\biggr]. 
\end{eqnarray*}
It is easy to see that \(\mathcal{K}({\tt 1})=\mathcal{L}({\tt 1})=0\).

Furthermore, we define highpass filters \(h_1,\cdots,h_6\) and \(h_1^d,\cdots,h_6^d\) as
\[\hat{h_j}(\ome):=\hat{h}(\ome)\overline{{l_j}(e^{2i\ome})},\quad \hat{h_j^d}(\ome):=\hat{h}(\ome)\overline{{k_j}(e^{2i\ome})},\quad j=1,2,3,4,5, 6,\]
and define highpass filters \(h_7,h_8,h_9,h_{10}\) as, with \(\nu_0=(0,0)\), \(\nu_1=(1,0)\), \(\nu_2=(0,1)\), and \(\nu_3=(1,1)\),
\[\hat{h_{6+1+m}}(\ome):=\frac{1}{2}e^{i\ome\cdot{\nu_m}} -\hat{h}(\ome)\overline{H_{\nu_m}(e^{2i\ome})},\quad m=0,1,2,3.\]
Then, \(\{h, h_1,\cdots, h_6, h_7,h_8,h_9,h_{10}\}\) and \(\{h, h_1^d,\cdots, h_6^d, h_7,h_8,h_9,h_{10}\}\) are the primal and dual wavelet filter banks, by Theorem~\ref{thm:main}.
\end{example}

\begin{example}
\label{ex:second}
Consider the two-dimensional quincunx case with the dilation matrix 
\[\dil:=\left[\begin{array}{ccc}
                1&&1\\
                1&&-1\\
                \end{array}\right].\] 
Then, \(|\det\dil|=2\). We choose \(\capset=\{(0,0),(1,0)\}\) and \(\dualset = \{(0,0),(\pi,\pi)\}\). 
Let \(h\) be a filter supported on \([-2,2]\times[-1,1]\) given by (cf. \cite{HJ})
 \[h                =\frac{\sqrt{2}}{32}\left[\begin{array}{ccccccccc}
                -1&&0&&2&&0&&-1\\
                 0&&8&&16&&8&&0\\
                -1&&0&&2&&0&&-1\end{array}\right].\] 
We next show that the SVP condition holds true in this case, with \(g\) chosen to be the same as \(h\).

Note that the \(z\)-transform \(H(z)\) of \(h\) is given as, for \(z=(z_1,z_2)\in \TT^2\), 
\begin{eqnarray*}
H(z)&=&\frac{\sqrt{2}}{2}+\frac{\sqrt{2}}{16}(4z_1+4z_1^{-1}+z_2+z_2^{-1})\\&-&\frac{\sqrt{2}}{32}(z_1^2z_2^{-1}+z_1^2z_2+z_1^{-2}z_2^{-1}+z_1^{-2}z_2).
\end{eqnarray*} 
Recall that \(h_{\nu}(m)=h(\dil m-\nu),\) \(\forall m \in \mathbb{Z}^2, \nu \in \capset\). Here, we have \( h_{(0,0)}\left(m\right)=h\left(m_1+m_2,m_1 -m_2\right)\) and \( h_{(1,0)}\left(m\right)=h\left(m_1+m_2-1,m_1 -m_2\right)\).

Thus, \(H_{(0,0)}(z)\) and \(H_{(1,0)}(z)\) can be computed as follows:
\begin{eqnarray*}
H_{(0,0)}(z)&=&\sum_{m \in \mathbb{Z}^2} h(m_1+m_2,m_1-m_2) z_1^{-m_1} z_2^{-m_2}=h(0,0)=\frac{\sqrt{2}}{2},
\\H_{(1,0)}(z)&=&\sum_{m \in \mathbb{Z}^2} h(m_1+m_2-1,m_1-m_2) z_1^{-m_1} z_2^{-m_2}\\
&=&\sum_{(\alp,\bet)\in \mathbb{Z}^2} h\left(\alp,\bet\right) z_1^{-(\alp+\bet+1)/2} z_2^{-(\alp-\bet+1)/2}\\
&=&\frac{\sqrt{2}}{32}\left(\right.- z_2-z_1+8+2z_1^{-1}+2z_2^{-1}\\&+&8 z_1^{-1} z_2^{-1}- z_1^{-2} z_2^{-1}-z_1^{-1} z_2^{-2}\left.\right).\end{eqnarray*} 
This gives the polyphase representation
\[
{\tt{H}}(z)=\frac{\sqrt{2}}{32}\biggr[16, \biggr(-z_2-z_1+8+2 z_1^{-1}+2 z_2^{-1}+8 z_1^{-1} z_2^{-1}- z_1^{-2} z_2^{-1}-z_1^{-1} z_2^{-2}\biggr)\biggr].
\]
Recalling our assumption that \(g=h\), we proceed with the computation
\begin{eqnarray*}
&&1-{\tt{H}}(z){\tt{H}}^\ast(z)\\&=&\frac{1}{2}-\frac{2}{32^2} \cdot \biggr(140+16 (z_1+z_2+z_1^{-1}+z_2^{-1})\\&-&4(z_1^2+z_2^2+z_1^{-2}+z_2^{-2})\\
&-&16(z_1^2 z_2+z_1 z_2^2+z_1^{-2} z_2^{-1}+z_1^{-1} z_2^{-2})+56(z_1z_2+z_1^{-1} z_2^{-1})\\
&+&2(z_1^2 z_2^2+z_1^{-2} z_2^{-2})\\&+&6 (z_1z^{-1}+z_1^{-1} z_2)+(z_1^3z_2+z_1 z_2^3+z_1^{-3} z_2^{-1}+z_1^{-1} z_2^{-3})\biggr),
\end{eqnarray*}
which gives that \(32^2(1-{\tt{H}}(z){\tt{H}}^\ast(z))\) is equal to
\begin{eqnarray*} 
&{-}&32\left(1-z_1^2 z_2\right)\left(1-z_1^{-2} z_2^{-1}\right) -32\left(1-z_1 z_2^2\right)\left(1-z_1^{-1} z_2^{-2}\right)\\&-&{8}\left(1-z_1^2\right)\left(1-z_1^{-2}\right)
-{8}\left(1-z_2^2\right)\left(1-z_2^{-2}\right)\\&+&{2}\left(1-z_1^3 z_2\right)\left(1-z_1^{-3} z_2^{-1}\right)+{4}\left(1-z_1^2 z_2^2\right)\left(1-z_1^{-2} z_2^{-2}\right)\\
&+&{2}\left(1-z_1 z_2^3\right)\left(1-z_1^{-1} z_2^{-3}\right)+{112}\left(1-z_1 z_2\right)\left(1-z_1^{-1} z_2^{-1}\right)\\
&+&32\left(1-z_1\right)\left(1-z_1^{-1}\right)+32\left(1-z_2\right) \left(1-z_2^{-1}\right)\\&+&{12}\left(1-z_1 z_2^{-1}\right)\left(1-z_1^{-1}z_2\right).
\end{eqnarray*}
Hence, we see that \(1-{\tt{H}}(z){\tt{H}}^\ast(z)=-\sum_{j=1}^{4} k_j(z)\overline{k_j(z)}+\sum_{j=5}^{11} k_j(z)\overline{k_j(z)},\)
where \([k_1(z), k_2(z), k_3(z), k_4(z)]\) is chosen as
\[ \frac{\sqrt{2}}{16}\biggr[2\left(1-z_1^2 z_2\right), 2\left(1-z_1 z_2^2\right),\left(1-z_1^2\right),\left(1-z_2^2\right)\biggr]\]
and \([k_5(z), k_6(z), k_7(z), k_8(z), k_9(z), k_{10}(z), k_{11}(z)]\) is chosen as
\begin{eqnarray*}
&{\,}&\biggr[\frac{\sqrt{2}}{32}\left(1-z_1^{3} z_2\right),\frac{2}{32}\left(1-z_1^{2} z_2^{2}\right),\frac{\sqrt{2}}{32}\left(1-z_1 z_2^{3}\right),\\
& &\frac{\sqrt{7}}{8}\left(1-z_1 z_2\right),\frac{\sqrt{2}}{8}\left(1-z_1\right),\frac{\sqrt{2}}{8}\left(1-z_2\right),\frac{\sqrt{3}}{16}\left(1-z_1 z_2^{-1}\right)\biggr].
\end{eqnarray*}
It is easy to see that they satisfy the vanishing condition \(k_j({\tt 1})=0\) for \(1\le j\le 11\). In other words, the SVP condition (\ref{eq:SVP}) is satisfied in this case.

Therefore, by Theorem~\ref{thm:main}, if we define highpass filters \(h_j\) for \(1\le j\le 11\) as
\[\hat{h_j}(\ome):=\frac{1}{\sqrt{2}}\overline{{k_j}(e^{i\dil^T\ome})}H (e^{i\ome}),\quad j=1,\cdots,11,\]
and define highpass filters \(h_{12}\) and \(h_{13}\) as, for \(m=0,1\) and \(\nu_0=(0,0)\), \(\nu_1=(1,0)\),
\[\hat{h_{11+1+m}}(\ome):=\frac{1}{\sqrt{2}}e^{i\ome\cdot{\nu_m}} -\frac{1}{\sqrt{2}}H (e^{i\ome})\overline{H_{\nu_m}(e^{i\dil^T\ome})},\]
then \[\{h, h_1,h_2,h_3, h_4, h_5,\cdots,h_{13}\}, \quad \{h, -h_1,-h_2,-h_3, -h_4, h_5,\cdots,h_{13}\}\] are the primal and dual wavelet filter banks. Note that in this case, we get the so-called quasi-tight wavelet filter bank, which is studied as a useful variant of tight wavelet filter banks, allowing \(-1\) multiplied to the highpass filters for dual wavelet filter bank (see, for example, \cite{diao}).
\end{example}

\begin{example}
\label{ex:third}
Consider the 1-D dyadic case, i.e. \(\dil=2\), and choose \(\capset=\{0,1\}\), \(\dualset=\{0,\pi\}\).
Let \(h\) be the lowpass filter associated with the centered Deslauriers-Dubuc interpolatory refinable function of order \(4\), supported on \([-3,3]\) (cf. the references \cite{SLP,ID}). Then the \(z\)-transform of \(h\) is
\[H(z)=\sqrt{2} \left(\frac{z+2+z^{-1}}{4}\right)^2 \left(\frac{-z+4-z^{-1}}{2}\right).\] 
Let us now show that the SVP condition (\ref{eq:SVP}) with \(g:=h\) is satisfied for this example as well.

Note that we have, for \(z\in\TT\),
\begin{equation*}
H(z)=\sum_{\nu\in\Gam}z^{\nu}H_\nu (z^2)=\frac{\sqrt{2}}{32}\left(-\frac{1}{z^3}+\frac{9}{z}+16+9z-z^3\right).
\end{equation*}
From this \(H(z)\), we get the following polyphase representaiton
\[{\tt H}(z)=\frac{\sqrt{2}}{32}\left[16,-\frac{1}{z^2}+\frac{9}{z}+9-z\right].\] 
Then, we have 
\[1-{\tt{H}}(z){\tt{H}}(z)^\ast=\frac{63}{512}(1-z)(1-z^{-1})-\frac{9}{256}(1-z^2)(1-z^{-2})+\frac{1}{512}(1-z^3)(1-z^{-3}).\]
Hence, 
\(
1-{\tt{H}}(z){\tt{H}}^\ast(z)=- k_1(z)\overline{k_1(z)}+\sum_{j=2}^{3} k_j(z)\overline{k_j(z)},
\)
where \(k_1(z), k_2(z), k_3(z)\) are
\[k_1(z)=\frac{3}{16}(1-z^2), \quad k_2(z)=\frac{3\sqrt{14}}{32}(1-z),\quad k_3(z)=\frac{\sqrt{2}}{32}(1-z^3),\]
and \(k_1({\tt 1})=k_2({\tt 1})=k_3({\tt 1})=0\), which implies the SVP condition is satisfied.

Therefore, by Theorem~\ref{thm:main}, if we define highpass filters \(h_1, h_2, h_3\) as
\[\hat{h_j}(\ome)=\frac{1}{\sqrt{2}}\overline{{k_j}(e^{2i\ome})}H (e^{i\ome}),\quad j=1,2,3,\]
and highpass filters \(h_4,h_5\) as
\[\hat{h_{3+1+m}}(\ome):=\frac{1}{\sqrt{2}}e^{i\ome\cdot{m}} -\frac{1}{\sqrt{2}}H (e^{i\ome})\overline{H_m(e^{2i\ome})}\,\quad m=0,1,\]
then \(\{h, h_1,h_2,h_3, h_4, h_5\}\) and \(\{h, -h_1,h_2,h_3, h_4, h_5\}\) are the primal and dual wavelet filter banks. In this case, as in Example~\ref{ex:second}, we get a quasi-tight wavelet filter bank \cite{diao}.

On the other hand, note that the filter \(h\) in this example satisfies the sub-QMF condition, i.e., \(1-{\tt{H}}(z){\tt{H}}^\ast(z)\ge 0\), \(\forall z\in \TT\), (see, for example,  \cite{SLP}), and thus there exists a Laurent polynomial \(p(z)\) such that \(1-{\tt{H}}(z){\tt{H}}^\ast(z)=|p(z)|^2\) by the Fej\'er-Riesz lemma \cite{Daub}. In fact, direct computation reveals that
\[p(z)=\frac{-2\sqrt{2}+\sqrt{6}}{32}+\frac{6\sqrt{2}-\sqrt{6}}{32}z+\frac{-6\sqrt{2}-\sqrt{6}}{32}z^2+\frac{2\sqrt{2}+\sqrt{6}}{32}z^3\]
satisfies the identity. Hence, by Result~\ref{result:SOStoUEP}, if we define highpass filter \(g_1\) as
\[\hat{g_1}(\ome):=\frac{1}{\sqrt{2}}\overline{{p}(e^{2i\ome})}H (e^{i\ome}),\]
then \(\{h, g_1,h_4,h_5\}\) is a tight wavelet filter bank. 

This wavelet filter bank is generally preferred because it is a tight wavelet filter bank and has fewer highpass filters. However, the process of obtaining the quasi-tight wavelet filter bank mentioned above is much simpler. Furthermore, the number of nonzero coefficients for the filter \(g_1\) is \(11\), while the numbers for the filters \(h_1,h_2\), and \(h_3\) are \(8, 6\), and \(8\), respectively. If filters with a smaller number of nonzero coefficients are preferred, the above quasi-tight wavelet filter bank may be a better choice. It is also clear that unlike this one-dimensional example, in high dimensions, determining whether a tight wavelet filter bank construction is feasible is highly nontrivial, as solving the SOS problem is required. 
\end{example}

%

\section*{Acknowledgments}
Part of the work was performed during the first author's visit to the Korea Institute for Advanced Study, Seoul 02455, Korea.

\end{document}